\documentclass[aps, prl,twocolumn,showpacs,superscriptaddress,floatfix]{revtex4-1} 
\usepackage{graphicx}
\graphicspath{{images}}
\usepackage{xcolor}
\usepackage{mathtools,amssymb,amsmath, amsthm,stmaryrd,thmtools,braket}
\usepackage{verbatim}

\usepackage{bm}
\usepackage{yhmath} 
\usepackage[inline]{asymptote}
\usepackage{cancel}
\usepackage{relsize}
\usepackage{array}

\newcommand{\ws}{\text{ }}
\newcommand{\expect}{{\rm I\kern-.3em E}}
\newcommand{\Var}{\mathrm{Var}\ws}
\newcommand{\snorm}[1]{||#1||_{\rm s}}
\newcommand{\norm}[1]{||#1||}
\newcommand{\abs}[1]{\left|#1\right|}
\let\vec\bm 
\newcommand{\vt}{\vec\theta}

\newcommand{\grad}{\nabla}
\newcommand{\expectation}[1]{\mathbb{E}\left[#1\right]}

\declaretheorem[name=Definition,Refname={Definition,Definitions}]{definition}
\declaretheorem[name=Theorem,Refname={Theorem,Theorems}]{thm}

\declaretheorem[name=Corollary,Refname={Corollary,Corollaries}]{cor}

\declaretheoremstyle[
    headfont=\bfseries, 
    notebraces={[}{]},
    bodyfont=\normalfont,
    headpunct={},
    postheadspace=\newline,
    spacebelow=\parsep,
    spaceabove=\parsep 
]{case}

\usepackage{xr-hyper}
\usepackage[colorlinks=true, urlcolor=blue,citecolor=blue,anchorcolor=blue]{hyperref}
\usepackage[capitalise]{cleveref}

\makeatletter
\newcommand*\l@sectioninfo{\@nodottedtocline{1}{1.5em}{2.3em}}
\newcommand*\l@subsectioninfo{\@nodottedtocline{2}{3.8em}{3.2em}}
\def\@nodottedtocline#1#2#3#4#5{%
  \ifnum #1>\c@tocdepth \else
    \vskip \z@ \@plus.2\p@
    {\leftskip #2\relax \rightskip \@tocrmarg \parfillskip -\rightskip
     \parindent #2\relax\@afterindenttrue
     \interlinepenalty\@M
     \leavevmode
     \@tempdima #3\relax
     \advance\leftskip \@tempdima \null\nobreak\hskip -\leftskip
     {#4}\nobreak
     \leaders\hbox{$\m@th
        \mkern \@dotsep mu\hbox{\,}\mkern \@dotsep
        mu$}\hfill
     \nobreak
     \hb@xt@\@pnumwidth{\hfil\normalfont \normalcolor }%
     \par}%
  \fi}

\makeatother

\def\sectioninfo#1{%
    \addcontentsline{toc}{sectioninfo}{%
    \noexpand\numberline{}\textit{#1}}%
}
\def\subsectioninfo#1{%
    \addcontentsline{toc}{subsectioninfo}{%
    \noexpand\numberline{}\textit{#1}}%
}

\begin{document}
\title{Optimal Measurement of Field Properties with Quantum Sensor Networks}
\author{Timothy Qian}
\affiliation{Joint Center for Quantum Information and Computer Science, NIST/University of Maryland College Park, Maryland 20742, USA}
\affiliation{Joint Quantum Institute, NIST/University of Maryland College Park, Maryland 20742, USA}
\affiliation{Montgomery Blair High School, Silver Spring, Maryland 20901, USA}
\author{Jacob Bringewatt}
\affiliation{Joint Center for Quantum Information and Computer Science, NIST/University of Maryland College Park, Maryland 20742, USA}
\affiliation{Joint Quantum Institute, NIST/University of Maryland College Park, Maryland 20742, USA}

\author{Igor Boettcher}
\affiliation{Joint Quantum Institute, NIST/University of Maryland College Park, Maryland 20742, USA}

\author{Przemyslaw Bienias}
\affiliation{Joint Center for Quantum Information and Computer Science, NIST/University of Maryland College Park, Maryland 20742, USA}
\affiliation{Joint Quantum Institute, NIST/University of Maryland College Park, Maryland 20742, USA}

\author{Alexey V. Gorshkov}
\affiliation{Joint Center for Quantum Information and Computer Science, NIST/University of Maryland College Park, Maryland 20742, USA}
\affiliation{Joint Quantum Institute, NIST/University of Maryland College Park, Maryland 20742, USA}

\date{\today}

\begin{abstract}
    We consider a quantum sensor network of qubit sensors coupled to a field $f(\vec{x};\vec{\theta})$ analytically parameterized by the vector of parameters $\vec\theta$. The qubit sensors are fixed at positions $\vec{x}_1,\dots,\vec{x}_d$. While the functional form of $f(\vec{x};\vec{\theta})$ is known, the parameters $\vec{\theta}$ are not. We derive saturable bounds on the precision of measuring an arbitrary analytic function $q(\vec{\theta})$ of these parameters  and construct the optimal protocols that achieve these bounds. Our results are obtained from a combination of techniques from quantum information theory and duality theorems for linear programming. They can be applied to many problems, including optimal placement of quantum sensors, field interpolation, and the measurement of functionals of parametrized fields.
\end{abstract}
\maketitle

It is well established that entangled probes in quantum metrology can be used to obtain more accurate measurements than unentangled probes \cite{bollinger1996optimal, huelga1997improvement, paris2009quantum, pezze2009entanglement, toth2012multipartite, zhang2014quantum}. In particular, while measurements of a single parameter using $d$ unentangled probes asymptotically obtain a mean squared error (MSE) from the true value of order $\mathcal{O}(1/d)$, using $d$ maximally entangled probes, each coupled independently to the parameter, one obtains an MSE of order $\mathcal{O}(1/d^2)$ -- the so-called Heisenberg limit \cite{wineland1992spin, bollinger1996optimal}. More recently, understanding the role of entanglement and generalizing this scaling advantage to the measurement of multiple parameters at once or functions of those parameters has been an area of keen interest \cite{gao2014bounds, zhang2014quantum, ragy2016compatibility, Eldredge2018, gessner2018sensitivity,proctor2018multiparameter, altenburg2018multi, zhuang2018distributed, Qian2019, gatto2019distributed, albarelli2019evaluating, sekatski2020optimal, sidhu2020geometric, guo2020distributed, oh2020optimal, zhuang2020distributed, rubio2020quantum} due to a wide array of practical applications \cite{spagnolo2012quantum, genoni2013optimal, humphreys2013quantum, yue2014quantum, baumgratz2016quantum, sidhu2017quantum, kok2017role}.
Importantly, optimal bounds and protocols have been derived for measuring analytic functions of independent parameters, each coupled to a qubit sensor in a so-called quantum sensor network \cite{Qian2019}. The problem of directly measuring a spatially dependent field of known form, possibly with extra noise sources, has also been considered \cite{sekatski2020optimal}.

In this Letter, we consider the following very general problem that is relevant for many technological applications of quantum sensor networks. A set of quantum sensors at positions $\{\vec{x}_1,\dots,\vec{x}_d\}$ is locally probing a physical field $f(\vec{x};\vec{\theta})$, which depends on a set of parameters $\vec{\theta}\in\mathbb{R}^k$, where we have used boldface to denote vectors. We assume that we know the functional form of $f(\vec{x};\vec{\theta})$, but we do not know the values of the parameters $\vec{\theta}$. For instance, these parameters may be the positions of several known charges, and $f(\vec x; \vec\theta)$ one of the components of the resulting electric field. Our objective is to measure a function of the parameters $q(\vec{\theta})$. This could be, for instance, the field value $q(\vec{\theta})=f(\vec{x}_0;\vec{\theta})$ at a position $\vec{x}_0$ without sensor, or the spatial average $q(\vec{\theta})=\int_R\mbox{d}\vec{x}\ f(\vec{x};\vec{\theta})$ over some region  $R$ of interest. In the following, we derive saturable bounds on the precision for measuring $q(\vec{\theta})$  using quantum entanglement. The setup is depicted in \cref{fig:sensornetwork}.

\begin{figure}[t]
    \centering
    \includegraphics[width=0.9\columnwidth]{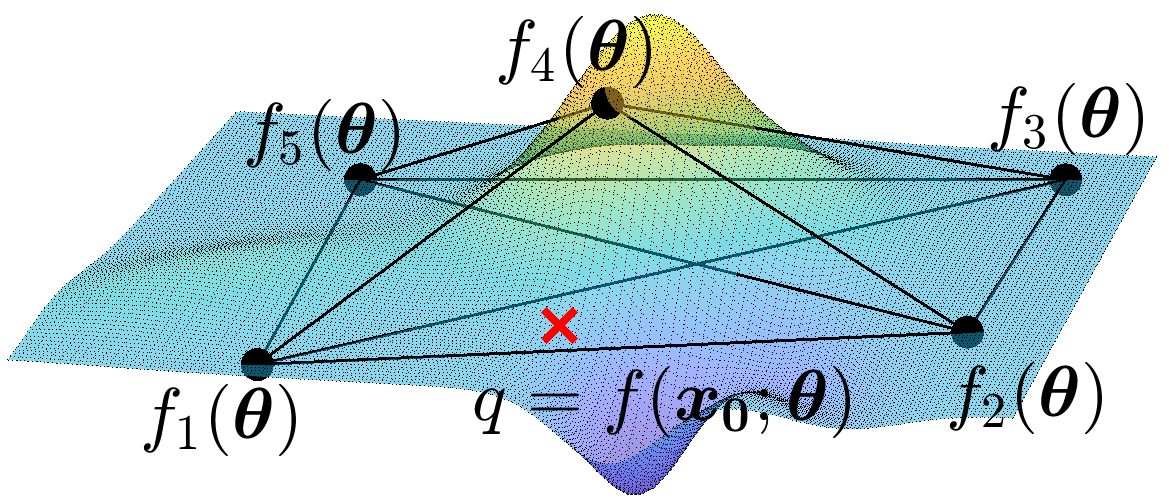}
    \caption{At each position $\vec{x}_i$ in the network, a quantum sensor (black dots) is coupled to a field $f(\vec{x};\vec{\theta})$, whose functional form is known, but the parameters $\vec{\theta}$ are not. The protocols presented here utilize entanglement to obtain the highest accuracy allowed by quantum mechanics in estimating the quantity $q(\vec{\theta})$. One example problem that our work solves is to estimate the field value  $q=f(\vec{x}_0;\vec{\theta})$ at a location $\vec{x}_0$ (red cross) without a sensor.}
    \label{fig:sensornetwork}
\end{figure}

As a more concrete example, consider a network of three quantum sensors that are locally coupled to a field $f(\vec{x};\theta_1,\theta_2)$ parametrized by $\vt=(\theta_1,\theta_2)$. The field amplitudes at the positions of the sensors shall be $f_1=\theta_1$, $f_2=\theta_2$, $f_3=\theta_1+\theta_2$, respectively, where we have introduced the shorthand notation $f_i(\vt)=f(\vec{x_i};\vt)$. Assume we want to measure the value of $q(\theta_1,\theta_2)=\theta_1$. One possible strategy is to simply use the first sensor to measure $f_1$. On the other hand, we could also measure $\frac{1}{2}(f_1-f_2+f_3)$, thereby potentially gaining accuracy by harnessing entanglement between the individual sensors. In fact, there are infinitely many variations of the second strategy, and we eventually expect some of them to be superior to the first strategy.

In contrast to previous work \cite{Qian2019}, where one considers estimating a given function $F(f_1,\dots,f_d)$ of independent local field amplitudes $f_1,\dots,f_d$, we consider here the problem of estimating a function of the parameters, $q(\theta_1,\dots,\theta_k$), instead. Due to the correlation of the local field amplitudes, there are many measurement strategies that need to be considered and compared in terms of accuracy. In this Letter, we determine the optimal protocol for this very general setup. Our work has a variety of important applications, including optimal spatial sensor placement and field interpolation.

\textit{Problem setup.---}We formally consider a quantum sensor network as a collection of $d$ quantum subsystems, called sensors, each associated with a Hilbert space $\mathcal{H}_i$ \cite{proctor2017networked, proctor2018multiparameter}. The full Hilbert space is $\mathcal{H}=\bigotimes_{i=1}^d\mathcal{H}_i$. We imprint a collection of field amplitudes $\vec {f}(\vt)=(f_1(\vt), \dots, f_d(\vt))^T$ onto a quantum state, represented by an initial density matrix $\rho_{\rm in}$, through the unitary evolution $\rho_{\rm f}=U(\vec{f})\rho_{\rm in} U(\vec{f})^\dagger$. Here, $\vt = (\theta_1, \dots, \theta_k)^T$ is a set of independent unknown parameters.
To be specific, we consider qubit sensors and a unitary evolution generated by the Hamiltonian
\begin{equation}\label{eq:H}
    \hat{H}=\hat{H}_{\rm c}(t)+\sum_{i=1}^d\frac{1}{2}f_i(\vt)\hat{\sigma}_i^{z},
\end{equation}
with 
$\hat{\sigma}_i^{x,y,z}$ the Pauli operators acting on qubit $i$
and $f_i(\vec{\theta})=f(\vec{x}_i,\vec{\theta})$ the local field amplitude at the position of the $i^{\mathrm{th}}$ sensor. Our results apply to more general quantum sensor networks (see Outlook).  The term  $\hat{H}_{\rm c}(t)$ is a time-dependent control Hamiltonian that we choose, which may include coupling to ancilla qubits. 

Our goal is to estimate a given function of the parameters $q(\vt)$. The estimate is based on measurements of the final state $\rho_{\rm f}$, specified by a set of operators $\{\hat\Pi_\xi\}$ that constitute a positive operator-valued measure (POVM) with $\int \mbox{d}\xi\ \hat{\Pi}_\xi=1$. We repeat this experiment many times and estimate the function of interest $q(\vt)$ via an estimator $\tilde{q}$ obtained from the data. On a more technical level, we assume that the sensor placements allow us to obtain an estimate of the true value of $\vec{\theta}$, which ensures the problem is solvable \footnote{Formally, we assume the ability to make an asymptotically (in  time $t$ per measurement run and in the number of measurement runs $\mu$)  unbiased, arbitrarily-small-variance estimate. See Supplemental Material for detailed definitions \cite{qian2020supplement}}.  This implies that 
the number $d$ of quantum sensors  should be larger than $k$. (See Outlook for cases where we can violate this assumption.)  
The choice of initial state $\rho_{\rm in}$,  control Hamiltonian $\hat{H}_{\rm c}(t)$, and POVM $\{\hat{\Pi}_\xi\}$ defines a \emph{protocol to estimate $q(\vec{\theta})$}.

Before proceeding, let us fix some notation. We treat $\vt$ as a stochastic variable and denote the true value of $\vt$ by $\vt'$. Thus $q(\vt)$ is again a stochastic quantity, whereas $q(\vt')$ is a specific number obtained by evaluating the function at the true value $\vt'$. We use indices $i,j=1,\dots,d$ to label quantum sensors and $m,n=1,\dots,k$ to label parameters. 

The MSE of the estimate $\tilde{q}$ from the true value $q(\vec{\theta}')$ is given by
\begin{equation}\label{eq:figofmerit}
\mathcal{M} = \expectation{\left(\tilde q-q(\vec\theta')\right)^2}=\Var \tilde q+\left(\expectation{\tilde q}-q(\vec\theta')\right)^2,
\end{equation}
where the first and second terms are the variance and estimate bias, respectively. We define the optimal protocol to measure $q(\vec{\theta}')$ as the one that minimizes $\mathcal{M}$ given a fixed amount of total time $t$. To determine the optimal protocol, we first derive lower bounds on $\mathcal{M}$ using techniques from quantum information theory. We then construct specific protocols that saturate these bounds.

\textit{MSE bound.---}In this section, we derive a saturable lower bound on $\mathcal{M}$ that can be achieved in time $t$ \footnote{Technically, to saturate our bounds, one requires $\mu$ measurements and thus a total time of $\mu t$ over all experimental runs. However, we avoid this technicality for notational clarity.}. To derive our bound, we begin with the following result on single-parameter estimation from Ref.~\cite{Boixo2007}. If the unitary evolution of the quantum state is controlled by a single parameter $q$, then \begin{equation}\label{eq:boixo}
    \mathcal{M} \geq \frac{1}{\mathcal{F}_Q} \geq \frac{1}{t^2\snorm{\hat{h}_q}^2},
\end{equation}
where $\mathcal{F}_Q$ is the quantum Fisher information, $\hat{h}_q=\partial \hat{H}/\partial q$ is the generator with $\gamma_{\max}$ ($\gamma_{\min}$) its largest (smallest) eigenvalue, and $\snorm{\hat{h}_q}=\gamma_{\max}-\gamma_{\min}$ is the seminorm of $\hat{h}_q$. The first inequality is the quantum Cram\'{e}r--Rao bound \cite{holevo2011probabilistic, helstrom1976quantum, braunstein1994statistical, braunstein1996generalized}. 

It is not obvious that Eq.\  (\ref{eq:boixo}) may be applied to the problem of estimating $q(\vt)$ as we have $k>1$ parameters controlling the evolution of the state. However, we circumvent this issue by considering an infinite set of imaginary scenarios, each corresponding to a choice of artificially fixing $k-1$ degrees of freedom and leaving only $q(\vt)$ free to vary. Under any such choice, our final quantum state depends on a single parameter, and we  can apply \cref{eq:boixo} to the imaginary scenario under consideration. 

We note that any such imaginary scenario requires giving ourselves information that we do not have in reality. However, additional information can only result in a lower value of $\mathcal{M}$. Therefore, any lower bound on $\mathcal{M}$ derived from any of the imaginary scenarios is also a lower bound for estimating the function $q(\vec{\theta})$. For a bound derived this way to be saturable, there must be some choice(s) of artificially fixing $k-1$ degrees of freedom that does not give us \emph{any} useful information about $q(\vt)$, and thus yields the sharpest possible bound. This is, in fact, the case. In our analysis below, the existence of such a choice becomes self-evident since we present a protocol that achieves the tightest bound. However, in the Supplemental Material, we prove that such a choice exists purely on information theoretic grounds \cite{qian2020supplement}. 

More formally, consider a basis $\{\vec\alpha_1, \vec\alpha_2, \cdots, \vec\alpha_k\}$ such that, without loss of generality, $\vec\alpha_1=\nabla q(\vec{\theta}')=:\vec\alpha$. We then consider any choice of the remaining basis vectors. For any such choice, let $\vec\alpha_n$ correspond to a function $q_n(\vt)=\vec\alpha_n\cdot\vt$. Therefore, if we consider a particular choice of basis, we are also considering a corresponding set of functions $\{q_1(\vt)=q(\vt), q_2(\vt), \cdots, q_k(\vt)\}$. We suppose we are given the values $\{q_n(\vt')\}_{n\geq 2}$, fixing $k-1$ degrees of freedom. The resulting problem is now determined by a single parameter,  and Eq.\ (\ref{eq:boixo}) applies. 

The derivative of $H$ with respect to $q$, while holding $q_2,\dots,q_k$ fixed, is
\begin{equation}
\label{eqn:generatorcomputation}
    \hat{h}_q= \frac{\partial\hat H}{\partial q}\Bigr|_{q_2,\dots,q_k} =\sum_{i = 1}^d \frac 12 (\nabla f_i(\vt')\cdot \vec\beta)\hat\sigma_i^z,
\end{equation}
where $\vec\beta = \left(\frac{\partial \theta_1}{\partial q}, \dots, \frac{\partial \theta_{k}}{\partial q} \right)|_{q_2,\dots,q_k}$. Using the chain rule, we find that $\vec{\beta}$ satisfies $\vec\alpha\cdot\vec\beta = 1$.

As we show formally in the Supplemental Material \cite{qian2020supplement}, every $\vec\beta\in \mathbb{R}^k$ in Eq.~(\ref{eqn:generatorcomputation}) corresponds to a valid choice of the $k-1$ dimensional subspace spanned by $\{\vec\alpha_n\}_{n\geq 2}$. 
 
Therefore, since $\hat{h}_q$ depends on 
$\{\vec\alpha_n\}_{n\geq 2}$ only through $\beta$,   
the tightest bound on $\mathcal{M}$ is found by optimizing over arbitrary choices of $\vec{\beta}$ subject to the constraint $\vec\alpha\cdot\vec\beta = 1$.

To formulate the corresponding optimization problem, define the matrix $G$ by
\begin{equation}
    \label{eqn:gradientmatrix}
    G_{im}(\vt') = \frac{\partial f_i}{\partial\theta_m}(\vt').
\end{equation}
We emphasize that $G$ depends on the true value of the parameters $\vt'$. 
Utilizing $\snorm{\frac 12 \hat\sigma_i^z} = 1$, we write the seminorm of $\hat{h}_q$ as
\begin{equation}
    \label{seminormevaluation}
    \snorm{\hat{h}_q} = \sum_{i = 1}^{d}\abs{ \nabla f_i(\vt')\cdot\vec\beta} = \norm{G(\vt')\vec\beta}_1,
\end{equation}
with $||\vec{x}||_1=\sum_{i=1}^d|x_i|$ the $L^1$ or Manhattan norm. Therefore, for any $\vec\beta$ satisfying $\vec\alpha\cdot\vec\beta = 1$, we have
\begin{equation}
    \label{eqn:generalmsebound}
    \mathcal{M} \geq  \frac1{t^2\snorm{\hat{h}_q}^2} = \frac 1{t^2\norm{G(\vt')\vec\beta}_1^2}.
\end{equation}
In order to obtain the sharpest bound, we must solve what we refer to as the bound problem for $G(\vt')$ and $\vec\alpha$:
\begin{center}
\begin{minipage}{0.4\textwidth}
\vspace{.5em}
\noindent \emph{\textbf{Bound problem:} Given a non-zero vector \  $\vec\alpha\in\mathbb{R}^k$ and a real $d\times k$ matrix $G$, compute $u=\max\limits_{\vec\beta}\frac{1}{\norm{G\vec \beta}_1}$ under the condition $\vec{\alpha}\cdot\vec{\beta}=1$.} 
\vspace{.5em}
\end{minipage}
\end{center}
\noindent This is a linear programming problem and can in general be solved in  time that is polynomial in $d$ and $k$ (see, e.g., Ref.\  \cite{jiang2020faster}). Hereafter, we refer to the resulting sharpest bound as ``the bound".

\textit{Optimal protocol.---}We now turn to the problem of providing a protocol that saturates this bound. For clarity of presentation, we develop this protocol in the case that both the field $\vec{f}(\vt)$  and the objective $q(\vt)$ are linear in the parameters $\vec \theta$; that is,
$\vec{f}(\vt) =G \vt$, with $\vt$-independent $G$,
and  $q(\vec\theta)=\vec\alpha\cdot\vec\theta$. However, the existence of an asymptotically optimal protocol can be proven in the more general case that $\vec f(\vec\theta)$ and $q(\vt)$ are analytic in the neighborhood of the true value $\vt'$ \cite{qian2020supplement}.

For the linear case, we propose an explicit protocol to measure $q$ and show that it saturates the bound and thus is optimal. The optimal protocol measures the linear combination
\begin{equation}\label{eq:lambda}
    \lambda(\vec{f})=\vec w\cdot \vec{f},
\end{equation}
where $\vec{f}$ is the vector of local field amplitudes. The vector $\vec w\in\mathbb{R}^d$ is chosen such that $\tilde{\lambda}(\vec{f})=\tilde{q}(\vt)$ is an unbiased estimator of $q(\vec\theta')$, and  will be optimized to saturate the bound. (We note that, for $d> k$, there are many choices of $\vec w$ that satisfy $\lambda=q$.)

For the estimator $\tilde{\lambda}$ to be unbiased, we must have $\expectation{\tilde q}=q(\vt')=\vec\alpha\cdot\vec\theta'$. This is achieved by choosing $\vec w$ to satisfy the \emph{consistency condition}
\begin{align}
  \label{eqn:consistency} 
  G^T \vec w = \vec\alpha.
\end{align}
Indeed, this implies
\begin{equation}\label{eqn:consistency2}
\begin{split}
    \expectation{\tilde q}&=\expectation{\vec w\cdot\vec{f}}=\left(G\vec\theta'\right)^T\vec w=\vec\theta'\cdot\left(G^T\vec w\right)=\vec\alpha\cdot\vec\theta'.
\end{split}
\end{equation}
We prove in the supplement that, under our assumption that we can estimate $\vt'$, \cref{eqn:consistency} may always be satisfied for some $\vec w$, and therefore our protocol is valid.

For any such choice of $\vec w$, we use the optimal linear protocol of Ref.~\cite{Eldredge2018} -- which for completeness, we summarize in the Supplemental Material \cite{qian2020supplement} -- to measure $\lambda(\vec{f})$.  
The variance obtained by this protocol is 
\begin{equation}
    \label{eqn:linearprotocolvariance}
    \Var \tilde q = \frac{\norm{\vec w}_\infty^2}{t^2},
\end{equation}
where $||\vec{w}||_\infty = \text{max}_i |w_i|$. Since we are dealing with an unbiased estimator, the MSE coincides with the variance of the estimator in \cref{eqn:linearprotocolvariance}. In order to find $\vec w$ with the lowest possible value of $\norm{\vec{w}}_\infty$ (i.e.\ the smallest variance), we must solve what we refer to as the protocol problem:
\begin{center}
\begin{minipage}{0.4\textwidth}
\vspace{.5em}
\emph{\textbf{Protocol problem:} Given a non-zero vector $\vec\alpha\in\mathbb{R}^k$ and a real $d\times k$ matrix $G$, compute $u'=\min\limits_{\vec{w}}\norm{\vec w}_\infty$ under the condition $G^T\vec w=\vec \alpha$.}
\vspace{.5em}
\end{minipage}
\end{center}
\noindent This, again, can be efficiently solved by generic linear programming algorithms \cite{panik2013linear,jiang2020faster} or special-purpose algorithms \cite{cadzow1973finite, cadzow1974efficient, abdelmalek1977minimum}. 

To show that the optimal protocol from solving this problem saturates the bound, we now show that the bound problem and protocol problem are equivalent in that $u=u'$. For this, we utilize the strong duality theorem for linear programming \cite{Luenberger1969, cadzow1973finite} \footnote{See Ref.~\cite{albarelli2019evaluating} for a quantum sensing use of this theorem in the context of evaluating the Holevo Cram\'{e}r--Rao bound.}. It states that, for linear programming problems like the protocol problem, there is a dual problem whose solution is identical to the original problem. In our case, we have the following dual problem:
\begin{center}
\begin{minipage}{0.4\textwidth}
\vspace{.5em}
\emph{\textbf{Dual protocol problem:} Given a non-zero vector $\vec\alpha\in\mathbb{R}^k$ and a real $d\times k$ matrix $G$, compute $u''=\max\limits_{\vec{v}}\ \vec\alpha\cdot\vec v$ under the condition $\norm{G\vec v}_1\leq 1$.}
\vspace{.5em}
\end{minipage}
\end{center}
\noindent The strong duality theorem then implies $u''=u'$. Additionally, there is a correspondence between the two solution vectors $\vec{w^0}$ and $\vec{v^0}$, so that, given the solution vector to one problem, we can find the solution vector to the other \cite{Luenberger1969, cadzow1973finite}. We now prove the following theorem.
\begin{thm}\label{thm:equiv}
Let $u$ and $u'$ be the solutions to the bound and protocol problems, respectively. Then $u=u'$.
\end{thm}

\begin{proof}
\noindent 
By the strong duality theorem, the solution of the dual protocol problem satisfies $u''=\max_{\vec v}\vec\alpha\cdot\vec v =u'$. Let the corresponding solution vector of the dual protocol problem be $\vec{v^0}$. Define $\vec\beta^0:=\vec{v^0}/u'$. We have $\vec\alpha\cdot\vec{\beta^0}=u'/u'=1$, thus $\vec{\beta^0}$ satisfies the constraint of the bound problem. To prove the theorem, we show that $u'\leq u$ and $u\leq u'$. On the one hand, provided $\norm{G\vec{\beta^0}}_1\neq 0$, the condition $\norm{G \vec {v^0}}_1\leq 1$ of the dual problem implies
\begin{equation}\label{eq:boundsat}
     u' \leq \frac{1}{\norm{G\vec{\beta^0}}_1} \leq \max_{\vec{\beta}} \frac{1}{||G \vec{\beta}||_1}=u.
\end{equation}
On the other hand, for \emph{any} $\vec \beta$ satisfying the constraint $\vec\alpha\cdot \vec\beta$ of the bound problem, and for the optimal $\vec w=\vec{w^0}$ of the protocol problem satisfying $\norm{\vec{ w^0}}_\infty=u'$, H\"older's inequality yields
\begin{align}
 \nonumber  &1 = \vec\alpha\cdot \vec\beta = (G^T \vec{w^0})^T \vec\beta = \vec{w^0} \cdot (G\vec\beta) \leq \norm{\vec{w^0}}_\infty \norm{G\vec\beta}_1\\
 \label{eqn:holderinequalitybound}    &\implies \frac 1{\norm{G\vec\beta}_1}\leq \norm{\vec{w^0}}_\infty=u'\ \text{for all}\ \vec \beta.
\end{align}
This shows that $u'\geq 1/||G\vec{\beta}||_1$ for all $\vec{\beta}$, thus $u'\geq u$, which completes the proof. As a byproduct, we learn from Eq.~(\ref{eq:boundsat}) that $\vec{\beta^0}$ maximizes $1/||G\vec{\beta}||_1$, and so is the solution vector of the bound problem. 
\end{proof}

\noindent Theorem \ref{thm:equiv} implies that the protocol measuring $\lambda$ with optimal $\vec w$ saturates the bound.

As an instructive example of how our three problems relate, we return to the toy model presented in the introduction. Consider three sensors coupled to local field amplitudes $f_1(\vt)=\theta_1$, $f_2(\vt)=\theta_2$, and $f_3(\vt)=\theta_1+\theta_2$. Our objective is $q(\vt)=\theta_1$, so $\vec\alpha=(1,0)^T$. We have 
\begin{equation}
    G^T=\begin{pmatrix}
    1 & 0 & 1 \\
    0 & 1 & 1
    \end{pmatrix}.
\end{equation}
First consider the bound problem. The constraint $\vec\alpha\cdot\vec\beta=1$ implies $\vec\beta=(1, b)^T$ with arbitrary $b$. The maximum of $1/\norm{G\vec\beta}_1$ is achieved for $\vec{\beta^0}=(1,0)^T$, yielding $u=1/2$. For the protocol problem, the constraint in Eq.~(\ref{eqn:consistency}) gives $w_1+w_2=1$ and $w_2+w_3=0$. The corresponding minimal value of $\norm{\vec w}_\infty$ is $u'=1/2$ for $\vec{w^0}=\left(\frac{1}{2}, -\frac{1}{2}, \frac{1}{2}\right)^T$. Finally, for the dual protocol problem, the constraint $\norm{G\vec v}_1\leq 1$ implies $|v_1|+|v_2|+|v_1+v_2|\leq 1$. The solution vector is $\vec{v^0}=(1/2,0)^T$, which yields $u''=\vec\alpha\cdot\vec{v^0}=1/2$. This explicit example demonstrates that $u=u'=u''$. Furthermore, as noted in the proof of Theorem \ref{thm:equiv}, $\vec{\beta^0}=\vec{v^0}/u'$.

\textit{Applications.---}Having derived optimal bounds and protocols saturating them, we now discuss some applications.  We begin by considering the same example as above and show that, remarkably, our results in this case indicate that the best entangled and best unentangled weighting strategies need not be the same. With or without entanglement, we estimate $q(\vt)=\theta_1$ by measuring a linear combination $\vec w\cdot\vec f$ with the constraints $w_1+w_3=1$, $w_2+w_3=0$. Without entanglement, our only option is to measure each component of $\vec f$ independently in parallel for time $t$, yielding a total MSE for $q(\vt)$ of $\norm{\vec w}_2^2/t^2$. In stark contrast, for the entangled case, the MSE is given by $\norm{\vec w}_\infty^2/t^2$. It is easy to see that minimizing the Euclidean and supremum norm of $\vec w$, subject to our constraints, does not yield the same result: Without entanglement, $\vec w=\left(\frac{2}{3}, -\frac{1}{3},\frac{1}{3} \right)^T$ is optimal, yielding an MSE of $\frac{2}{3t^2}$. With entanglement, $\vec w=\left(\frac{1}{2}, -\frac{1}{2}, \frac{1}{2}\right)^T$ is optimal, with MSE of $\frac{1}{4t^2}$. 
This simple example shows that, to achieve the optimal result with entanglement, one cannot in general use the weights $\vec w$ that are optimal without entanglement.  

Our results are practically relevant for any situation where one knows the functional form of the field of interest $f(\vec{x};\vec{\theta})$ and seeks to determine some quantity dependent on the parameters of that field. Examples include functionals of the form $q(\vec{\theta})=\int_R\mbox{d}\vec{x}\ k(\vec{x}) f(\vec{x};\theta)$ with any kernel $k(\vec{x})$ and region of integration $R$. The examples from the introduction  
correspond to $k(\vec{x})=\delta(\vec{x}-\vec{x}_0)$ and $k(\vec{x})=1$. Since the $\vec{\theta}$-dependence of $f(\vec{x},\vec{\theta})$ is analytic, this amounts to evaluating an analytic function $q(\vt)$. 

Our findings are also relevant for determining the optimal placement of sensors in space, i.e. determining the best locations $\vec{x}_1,\dots,\vec{x}_d$ in the control space $X$ in which they reside. For example, if the sensors are confined to a plane, then $X=\mathbb{R}^2$. This problem clearly consists of two parts: (1) evaluating the best possible MSE for any chosen set of sensor locations and (2) optimizing the result over possible locations. The MSE amounts to the cost function in usual optimization problems. Our results solve this first part as it would be used in the inner loop of a numerical optimization algorithm. The full problem, involving also the second part, is a high dimensional optimization in a space of dimension $d\times\mathrm{dim}(X)$. Therefore, in general, one expects that finding the global optimal placement could be quite challenging. However, even finding a local optimum in this space is clearly of practical use. 

\textit{Outlook.---}While we assumed that we can obtain an individual estimate of the true value $\vt'$ of the parameters, one could imagine situations where this assumption is not satisfied. Some such systems are underdetermined and not uniquely solvable, but in some cases we can reparametrize $\vt\rightarrow \vt^*$ in order to satisfy the assumption. For example, if two parameters in the initial parametrization always appear as a product $\theta_1\theta_2$ in both $f$ and $q$, we cannot individually estimate $\theta_1$ or $\theta_2$. However, we can reparametrize $\theta_1\theta_2\rightarrow \theta_1^*$ and thus satisfy our initial assumption.

Our work applies to physical settings beyond qubit sensors -- that is, any situation where \cref{eq:boixo}, may be applied our results should hold, provided we use the corresponding seminorm for the particular coupling. One example is using a collection of $d$ Mach-Zehnder interferometers where the role of local fields is played by interferometer phases \cite{holland1993interferometric, kim1998influence, devi2009quantum, dinani2016quantum, ge2018}. Here the limiting resource is the number of photons $N$ available to distribute among interferometers and not the total time $t$. The optimal variance for measuring a linear combination of local field values in this setting is conjectured in Ref.~\cite{Eldredge2018}. Under the assumption that this conjecture is correct, we may replace Eq.\ (\ref{eqn:linearprotocolvariance}) with $
    \mathcal{M} = \frac{\norm{\vec w}_\infty^2}{N^2}$
and otherwise everything remains the same as the qubit sensor case. One could also consider the entanglement-enhanced continuous-variable protocol of  Ref.~\cite{zhuang2018distributed} for measuring linear combinations of field-quadrature displacements. A variation of this protocol has been experimentally implemented in Ref.~\cite{guo2020distributed}.
We expect our bound and protocol 
could be extended to all the scenarios just described or even to the hybrid case where some local fields couple to qubits, some to Mach-Zehnder interferometers, and some to field quadratures. The ultimate attainable limit in such physical settings remains an open question, however.

One could 
consider 
the case $d<k$ provided the $d$ sensors are not required to be at fixed locations. For instance, if one had access to continuously movable sensors in a 1D control space $X$, by the Riesz representation theorem \cite{Luenberger1969}, one could encode any linear functional of $f(x; \vec\theta)$ by moving the sensors according to a particular corresponding velocity schedule. As a simple example, one can consider evaluating the integral of some function of (one component of) a magnetic field over one-dimensional physical space by moving a qubit sensor through the field and measuring the accumulated phase. One could also consider variations of this work in the context of 
semiparametric estimation \cite{tsang2020quantum}. We leave further exploration of such schemes to future work.

We thank Pradeep Niroula for discussions. We acknowledge funding by ARL CDQI, NSF PFC at JQI, AFOSR MURI, AFOSR, ARO MURI, NSF PFCQC program, DoE ASCR Accelerated Research in Quantum Computing program (award No.~DE-SC0020312), the DoE ASCR Quantum Testbed Pathfinder program (award No.~DE-SC0019040), and the U.S. Department of Energy Award No.~DE-SC0019449. J.B. acknowledges support by the U.S. DoE, Office of Science, DoE ASCR, DoE CSGF (award No.~DE-SC0019323). 

\let\oldaddcontentsline\addcontentsline
\renewcommand{\addcontentsline}[3]{}

\let\addcontentsline\oldaddcontentsline

\onecolumngrid
\vspace{2em}
\begin{center}
\textbf{Supplemental Material for ``Optimal Measurement of Field Properties with Quantum Sensor Networks''} 
\end{center}
\setcounter{secnumdepth}{1}
\setcounter{thm}{0}
\renewcommand{\thethm}{S.\arabic{thm}}
\renewcommand{\thelemma}{S.\arabic{lemma}}
\renewcommand{\thecor}{S.\arabic{cor}}
\setcounter{equation}{0}
\renewcommand{\theequation}{S.\arabic{equation}}
\renewcommand{\bibnumfmt}[1]{[S#1]}
\renewcommand{\citenumfont}[1]{S#1}
{
  \hypersetup{linkcolor=black}
  \tableofcontents
}

\section{Justification of using single-parameter bound}
\sectioninfo{We show that, by artificially fixing $k-1$ degrees of freedom, we can use the single-parameter bound in Eq.\ (\ref{eq:boixo}) of the main text. We show that any choice of \texorpdfstring{$\vec\beta\in\mathbb{R}^k$}{filler} satisfying \texorpdfstring{ $\vec\alpha\cdot\vec\beta=1$}{filler} corresponds to a valid choice of the $k-1$ dimensional subspace spanned by \texorpdfstring{ $\{\vec\alpha_n\}_{n=\{2,\dots,k\}}$}{filler}. We provide an information-theoretic proof that there is a choice of such a subspace that ensures we obtain no information about \texorpdfstring{ $q(\vt')$}{filler}.}

In this section, we elaborate on artificially fixing $k-1$ degrees of freedom in order to use the single-parameter bound in Eq.\ (\ref{eq:boixo}) in the main text. We begin by showing that any choice of $\vec\beta\in\mathbb{R}^k$ satisfying $\vec\alpha\cdot\vec\beta=1$ picks out a valid choice of a $k-1$ dimensional subspace that  $\{\vec\alpha_n\}_{n=\{2,\dots,k\}}$ must span such that the full set $\{\vec\alpha_n\}_{n=\{1,\dots,k\}}$ is a valid basis. 

We begin by noting that,  formally, the basis of vectors $\{\vec\alpha_n\}_{n=\{1,\dots,k\}}$ corresponds to the rows of the Jacobian matrix $J = [\vec{\alpha_1}, \dots, \vec{\alpha_k}]^T$ of the coordinate transformation between $\vec\theta'$ and $\vec q=(q_1(\vt'), q_2(\vt'), \cdots, q_k(\vt') )^T$. Further, there exists a dual basis of vectors  $\{\vec\beta_n\}_{n=\{1,\dots,k\}}$ corresponding to the columns of the inverse  Jacobian matrix $J^{-1} = [\vec{\beta_1}, \vec{\beta_2}, \dots,\vec{\beta_k}]$. Therefore, $\vec{\alpha_n}\cdot\vec{\beta_m}=\delta_{nm}$ since $JJ^{-1} = I$. In particular, we have
\begin{equation}
    \label{eqn:basisconditionagain}
    \vec{\alpha_1}\cdot\vec{\beta_1} = \vec\alpha\cdot\vec\beta = 1.
\end{equation}
Recall that $\vec\alpha=\grad q(\vt')$ is fixed by the quantity we desire to measure. Furthermore, we assume without loss of generality that $q_1(\vt)=q(\vt)$. Now suppose we pick \emph{any} $\vec\beta\in\mathbb{R}^k$ satisfying Eq.\ (\ref{eqn:basisconditionagain}). If there is a valid basis $\{\vec\alpha_n\}_{n=\{1,\dots,k\}}$ corresponding to this choice, we require that the $k-1$ vectors $\{\vec\alpha_n\}_{n=\{2,\dots,k\}}$ span the orthogonal complement of $\vec\beta$. Furthermore, we require that these vectors be independent of $\vec\alpha$. This is clearly true for any valid basis $\{\vec\alpha_n\}_{n=\{2,\dots,k\}}$ for the orthogonal complement of $\vec\beta$ as $\vec\alpha$ is not in this subspace via Eq.\ (\ref{eqn:basisconditionagain}). Therefore, we have reduced the problem to that of picking the optimal choice of $\vec\beta\in\mathbb{R}^k$. 

We now show via information theoretic arguments that there is such a choice of $\vec\beta$ that gives us no useful information about $q$ and that therefore the sharpest bound obtained by optimizing over $\vec\beta$ is in fact saturable. 
Let $\mathcal{F}(\vec\theta)=(\vec{\ell_1}, \vec{\ell_2}, \cdots, \vec{\ell_k})=(\vec{\ell_1}, \vec{\ell_2}, \cdots, \vec{\ell_k})^T$ be the Fisher information matrix with respect to the parameters $\vec\theta$, where we have explicitly indicated it is symmetric. Then we may use the previously defined Jacobian to obtain the Fisher information matrix with respect to $\vec q=(q_1(\vt')=q(\vt'), q_2(\vt'), \cdots, q_k(\vt') )^T$:
\begin{equation}
    \mathcal{F}(\vec q)=(J^{-1})^T\mathcal{F}(\vec\theta')J^{-1}.
\end{equation}
We note that, if $\mathcal{F}(\vec q)_{1n}=\mathcal{F}(\vec q)_{n1}=0$ for all $n\neq 1$,  then there is no information about the desired $q_1(\vec\theta')$ in the other $q_{n\neq 1}(\vec\theta')$. Therefore, if our bound is saturable, it must be possible to construct such an $\mathcal{F}(\vec q)$.

Let $\vec\alpha_1=\vec\alpha=(a_1, \cdots, a_k)^T$. Since we know there is a protocol  saturating our bound and since we know what it is, we propose the ansatz
\begin{equation}\label{eq:ansatz}
\vec{\beta_1}=\left(\frac{1}{a_1}, 0,\cdots, 0\right)^T.
\end{equation}
Eq.\ (\ref{eq:ansatz}) clearly satisfies $\vec\alpha_1\cdot\vec\beta_1=1$. Furthermore, we pick some choice of remaining basis vectors $\{\vec\alpha_n\}_{n=\{2,\dots,k\}}$ such that
\begin{equation}
\vec\alpha_{n\neq 1}^T=(0, \vec{v_n}^T),
\end{equation}
which satisfy $\vec\alpha_{n\neq1}\cdot\vec\beta_1=0$. Define the $(k-1)\times(k-1)$-dimensional matrix $V=(\vec v_2, \vec v_3, \cdots, \vec v_k)^T$. Then define $U^T=V^{-1}$. Therefore, letting $U=(\vec u_2, \vec u_3, \cdots, \vec u_k)^T$, we have $\vec u_n\cdot \vec v_m=\delta_{mn}$. Defining $\vec a=(a_2, a_3, \cdots, a_k)^T$, we then can pick 
\begin{equation}
\vec\beta_{n\neq 1}^T=\left(\frac{-\vec{u_n}\cdot\vec a}{a_1}, \vec{u_n}^T\right ),
\end{equation}
which clearly satisfies $\vec\alpha_n\cdot\vec\beta_m=\delta_{nm}$. We then have 
\begin{equation}
    \mathcal{F}(\vec q)_{n1}=\vec\beta_n^TF(\vt')\vec\beta_1= \vec\beta_n^T\frac{\vec{\ell_1}}{a_1}=\frac{\vec{\ell_1}\cdot\vec{\beta_n}}{a_1}=\mathcal{F}(\vec q)_{1n}.
\end{equation}
The above equation implies 
\begin{equation}
    \mathcal{F}(\vec q)_{11}= \frac{\ell_{11}}{a_1^2}, 
\end{equation}
where we let $\ell_{11}$ denote the first component of $\vec{\ell_1}$. Furthermore, if our choice of basis is to give us no information, we must have, for $n\neq 1$,
\begin{equation}
    \mathcal{F}(\vec q)_{1n}=\vec{u_n}\cdot\left(-\frac{\ell_{11}}{a_1}\vec a+\vec{\ell_1}'\right)=0,
\end{equation}
where $\vec{\ell_1}'=(\ell_{12}, \ell_{13}, \cdots, \ell_{1k})^T$. In other words, this gives the off-diagonal elements zero. It is impossible to have $k$ linearly independent vectors $\vec{u}_n$ all orthogonal to $\left(-\frac{\ell_{11}}{a_1}\vec a+\vec{\ell_1}'\right)$ in a $k$ dimensional space, so we demand
\begin{equation}
    -\frac{\ell_{11}}{a_1}\vec a+\vec{\ell_1}'=0.
\end{equation}
We note that, if $\mathcal{F}(\vec\theta)$ is diagonal,  this is impossible to satisfy. However, if $\vec{\ell_1}\propto\vec\alpha_1$, this is satisfied, which is in fact what is done in the linear protocol from Ref.~\cite{Eldredge2018} that we use as a subroutine in our protocol. Also note that $a_1$ must be the maximum-magnitude element of $\vec\alpha$ for this to be satisfiable due to the properties of the Fisher information matrix (namely $\ell_{11}\geq\ell_{1n}$ for $n\neq 1$). Without loss of generality, we let $a_1$ be this maximum-value element as the order of indexing our sensors is arbitrary. Therefore, we see that, by insisting that fixing $k-1$ degrees of freedom gives us no useful information, the protocol in Ref.~\cite{Eldredge2018} emerges naturally. 

We note that one can find a somewhat related argument regarding the results of Ref. \cite{Eldredge2018} in Ref. \cite{proctor2017networked}. 

\section{Proof of validity of the consistency condition}
\sectioninfo{We collect some definitions for estimators and prove that, provided we can estimate \texorpdfstring{ $\vt'$}{filler}, the consistency condition \texorpdfstring{$G^T({\vec\theta'})\vec w =\vec\alpha$ }{filler} [Eq.\ (\ref{eqn:consistency}) in the main text] is satisfied for some \texorpdfstring{ $\vec w$}{filler}. }

Here we prove that, provided we can estimate $\vt'$, the consistency condition $G^T({\vec\theta'})\vec w =\vec\alpha$ [Eq.\ (\ref{eqn:consistency}) in the main text] is satisfied for some $\vec w$. We use this result in Sec.~\ref{appendix:convergence} to prove that using $\tilde\vt$ (instead of $\vt'$) in the second step of the analytic-function-case protocol of Sec.~\ref{app:analyticprotocol}  induces negligible errors.

We begin by recalling a standard definition. 
\begin{definition}
An asymptotically unbiased estimator $\tilde\vt$ of $\vt$ is one that asymptotically (in time $t$ and the number of measurements $\mu$) has $\expect[\tilde\vt] = \vt$.
\end{definition}
\noindent We can now prove the following theorems.
\begin{thm}\label{thm:orthogonalcomplement}
If we can make an asymptotically unbiased estimate of a function $q(\vt')$ with arbitrarily small variance, then all $\vec\beta$ in the null space of $G(\vec\theta')$ lie in the orthogonal complement of $\vec\alpha = \nabla q(\vt')$.
\end{thm}
\begin{proof}
Proceeding by contradiction, suppose we have a $\vec\beta$ satisfying $G(\vt')\vec\beta = \vec 0$ [\textit{i.e.}~$\vec\beta$ is in the null space of $G(\vec\theta')$] and $\vec\alpha \cdot\vec\beta \neq 0$ (\textit{i.e.}~$\vec\beta$ is \emph{not} in the orthogonal complement of $\vec\alpha$). We can scale $\vec\beta$ by a constant to force $\vec\alpha \cdot \vec\beta = 1$ and maintain $G(\vt')\vec\beta = 0$. According to the bound in \cref{eqn:generalmsebound} of the main text, the MSE of any estimator of $q(\vt')$ then approaches infinity. Thus we can't make an asymptotically unbiased estimate with arbitrarily small variance, a contradiction. 
\end{proof}

\begin{thm}\label{thm:consistency}
If we can make an asymptotically unbiased estimate of $q(\vt')$ with arbitrarily small variance, then $G^T(\vec\theta')\vec w=\vec\alpha = \nabla q(\vt')$ is consistent.
\end{thm}
\begin{proof}
Theorem \ref{thm:orthogonalcomplement} implies that $\vec\alpha$ lies in the column space of $G^T(\vec\theta')$, as the null space of $G(\vt')$ is the orthogonal complement of the column space of $G^T(\vt')$. Thus the system $G^T(\vt')\vec w = \vec\alpha$ is consistent.  
\end{proof}
\begin{cor}\label{cor:singularvalues}
If we can make asymptotically unbiased estimates of  $\vec\theta'$ with arbitrarily small variance, then $G^T(\vec\theta')$ is full rank.
\end{cor}
\begin{proof}
If we can make an estimate of $\vec\theta'$, we can think of this as making an estimate of $q(\vt')=\theta'_i$ for any $i$. Therefore, by Thm. \ref{thm:orthogonalcomplement}, we have that $G^T(\vec\theta')\vec w = \vec e_i$ is consistent for any element $\vec e_i$ in the standard basis of $\mathbb{R}^k$, which implies that $G^T(\vec\theta')$ is full rank.
\end{proof}

\section{Optimal protocol: case of analytic functions}\label{app:analyticprotocol}
\sectioninfo{We describe a protocol that saturates our bounds when \texorpdfstring{$\vec f(\vec\theta)$}{filler}  and \texorpdfstring{$q(\vt)$ }{filler}   are analytic in the neighborhood of the true value \texorpdfstring{$\vt'$}{filler}.}

We use the results of the main text to generalize our optimal protocol to the case where both $\vec f(\vec\theta)$ and $q(\vt)$ are analytic in the neighborhood of the true value $\vt'$. Given a total time $t$, we consider a two-step protocol that extends the approach of Ref. \cite{Qian2019}. In the first step, we spend time $t_1=t^p$ with $p\in(1/2,1)$ to obtain an initial estimate $\tilde{\vec\theta}$ of the true value $\vt'$. We then linearize $q$ about $\tilde\vt$ to obtain
\begin{align}
\begin{split}
  \label{eq15}  q(\vt) \approx q(\tilde{\vt}) + \nabla q(\tilde{\vt})\cdot(\vt-\tilde{\vt}) =: \tilde{\vec{\alpha}}\cdot \vt + K,
\end{split}
\end{align}
where $\tilde{\vec \alpha}=\nabla q(\tilde{\vt})$ and $K$ is a constant with respect to $\vec{\theta}$. We will show in Sec.\ \ref{appendix:protocolasymptotics} that the error introduced by this approximation is negligible
if $\tilde{\vt}$ can be estimated with MSE $\mathcal{O}(1/t_1^2)$ in time $t_1$ (as can be done via phase estimation procedures like in Ref. \cite{kimmel2015robust}). After having obtained $\tilde{\vt}$ in the first step, we can compute $K$. 

In the second step, we estimate the remaining linear term $\tilde{\vec\alpha}\cdot \vt$ in Eq.\ (\ref{eq15}) in time $t_2=t-t_1$. Define $G(\tilde{\vec{\theta}})$ as in Eq.\  (\ref{eqn:gradientmatrix}) of the main text with $\vec{\theta}'\to \tilde{\vec{\theta}}$. Then following the procedure of the linear case protocol in the main text, we measure a linear function $\lambda$ such that the corresponding estimate $\tilde\lambda$ is an asymptotically unbiased estimate of $\tilde{\vec \alpha}\cdot \vec\theta$. In particular, here we have $\lambda = \tilde{\vec{w}}\cdot(\vec{f}(\vt)-\vec{C})$, where the constant vector $\vec{C}$ is chosen in such a way that $\vec{f}(\vt)-\vec{C}=G(\tilde{\vt})\vt +\mathcal{O}(\vec\Delta)$, with $\vec\Delta:=\tilde\vt-\vt'$ and $\tilde{\vec w}$ a vector of weights that we still need to choose. With $\vec C$ defined in this way, we linearize $f_i(\vt)$ about $\tilde\vt$ and obtain
\begin{equation}\label{eqn:lambdageneral}
   \lambda \approx \tilde{\vec w}\cdot (G(\tilde\vt)\vt) = (G(\tilde\vt)^T\tilde{\vec w}) \cdot \vt.
\end{equation}
Similar to the linear-case protocol in the main text, we ensure that $\tilde{\lambda}$ estimates $q$ by choosing $\tilde{\vec w}$ to satisfy  $G(\tilde\vt)^T\tilde{\vec w} = \tilde{\vec \alpha}$, which we show in \cref{appendix:convergence} to be a consistent system of equations. We then solve the corresponding protocol problem to obtain the optimal vector $\tilde{\vec w}$, given $G(\tilde{\vec\theta})$ and $\tilde{\vec\alpha}$.

Combining steps one and two yields an estimator for $q(\vec{\theta})$. In \cref{appendix:protocolasymptotics}, we show that the MSE for such a protocol is asymptotically equal to the linear case in \cref{eqn:linearprotocolvariance}. 
The crucial point is that the process of linearizing $\vec f$ and $q(\vt)$ about $\tilde \vt$ introduces asymptotically negligible corrections compared to linearizing about the true value $\vt'$. Consequently, asymptotically, 
\begin{equation}\label{eq:protocolfinal}
    \mathcal{M}\sim\frac{\norm{\tilde{\vec w}}_\infty}{t_2^2} \sim \frac{\norm{\vec w}_\infty}{t^2},
\end{equation}
where we have used $t_2\sim t$, and $\vec w$ is the optimal weight vector obtained from the protocol problem for $G = G(\vt')$ and $\vec\alpha = \nabla q(\vt')$. Referring to Eq.\ (\ref{eqn:gradientmatrix}) and the preceding discussion in the main text, we recall that $G$ and $\vec\alpha$ defined this way are precisely the appropriate input to the bound problem in order to obtain our ultimate MSE bound. 
As, asymptotically, our protocol for an analytic objective and field yields an MSE equivalent to the fully linear case, the same proofs as in the truly linear case hold,  and the asymptotic bound obtained by solving the protocol problem in Eq.\ (\ref{eq:protocolfinal}) is equivalent to the sharpest bound obtained by solving the corresponding bound problem -- therefore the protocol is asymptotically optimal.

\section{Proofs on estimate asymptotics}\label{appendix:convergence}
\sectioninfo{We prove that, in the context of the general two-step protocol of \cref{app:analyticprotocol}, using the estimate \texorpdfstring{$\tilde{\vt}$}{filler} obtained from the first step of the protocol, as opposed to the unknown true value \texorpdfstring{$\vt'$}{filler}, asymptotically yields negligible errors in the determination of the weight vector \texorpdfstring{$\vec w$}{filler}, which is the solution of the protocol problem. } 

In this section, we prove that, in the two-step protocol of \cref{app:analyticprotocol}, using the estimate $\tilde\vt$ obtained from the first step of the protocol, as opposed to the unknown true value $\vt'$,  asymptotically yields negligible errors when compared to the determination of the weight vector  $\vec w$ that is the solution of the protocol problem.

Recall that we use time $t_1=t^{p}$ for $1/2<p<1$ on the first step of the protocol to obtain an estimate of each $\vt_i$ with MSE $\mathcal O\left(\frac{1}{t_1^2}\right)$. We then spend time $t_2=t-t_1$ estimating $q(\vt')$ via a linearization of $q(\vt)$ about our estimate $q(\tilde\vt)$ with a weighted linear (in $\vt$) protocol. We begin by assuming that our initial estimate $\tilde\vt$ satisfies
\begin{equation}\label{eq:errorsize}
 \norm{\tilde\vt-\vt'}=\norm{\vec\Delta}\leq \delta 
\end{equation} 
for some fixed positive real $\delta$, where we defined $\vec\Delta:=\tilde\vt-\vt'$. From here on, norms without subscripts denote the Euclidean norm. This means we assume $\tilde\vt$ lies within or on a ball of radius $\delta$ of the true value $\vt'$ in the parameter space of $\vt$'s. 
Recall that we also require that both $q(\vt)$ and $f_i(\vt)$ $\forall i$ are analytic within this ball for some $\delta$. Crucially, asymptotically in time $t$, we may make $\delta$ an arbitrarily small fixed positive number. That is, as the total time $t\longrightarrow \infty$, the time spent obtaining our estimate $t_1=t^p\longrightarrow\infty$, and therefore the MSE of our estimate $\tilde\vt$ goes to zero. 

We now prove the following theorem, which guarantees that, asymptotically, $G(\tilde\vt)$, as defined in Eq.\ (\ref{eqn:gradientmatrix}) of the main text with $\vt'\rightarrow \tilde\vt$, has full rank.
\begin{thm}\label{thm:consistency}
Given a $d\times k$ matrix 
\[
G(\vt) = \begin{pmatrix} \frac{\partial f_1(\vt)}{\partial\theta_1} & \dots & \frac{\partial f_1(\vt)}{\partial\theta_{k}} \\
\vdots & \ddots & \vdots\\
\frac{\partial f_d(\vt)}{\partial\theta_1} & \dots & \frac{\partial f_d(\vt)}{\partial\theta_{k}}
\end{pmatrix}
\]
and an estimate $\tilde\vt$, asymptotically for $t\longrightarrow \infty$, $G(\tilde\vt)$ has full rank.
\end{thm}
\begin{proof}
From Theorem \ref{cor:singularvalues}, we know that all singular values of $G(\vt')$ are nonzero, and thus the matrix has full rank. Let $P$ be a perturbation matrix such that $G(\tilde\vt)=G(\vt')+P$. We expand $f_i(\tilde\vt)$ about $\vt'$ as
\begin{equation}\label{eq:phiexpansion}
    f_i(\tilde\vt)=f_i(\vt')+\nabla f_i(\vt')\cdot\vec\Delta+\dots
\end{equation}
Thus,
\begin{equation}
\begin{split}
    G(\tilde\vt)_{in}=\frac{\partial 
    f_i(\tilde\vt)}{\partial\theta_n}&=\frac{\partial f_i(\vt')}{\partial\theta_n}+\mathcal{O}(\vec\Delta)\\
    &=G(\vt')_{in}+\mathcal{O}(\vec\Delta),
\end{split}
\end{equation}
which implies $P_{in}=\mathcal{O}(\vec\Delta)$ and therefore, as $\dim(\vec\Delta)=k=\mathcal{O}(1)$, $\norm{P}=\mathcal{O}(\vec\Delta)$. It is a well-known result in matrix perturbation theory (see e.g.~Ref.~\cite{dahleh2004lectures}) that, if $\norm{P}<\sigma$, where $\sigma$ is the minimum singular value of $G(\vt')$, then $G(\tilde\vt)$ has the same rank as $G(\vt')$, i.e.~full rank. Since  asymptotically we can make $\delta$ arbitrarily small in \cref{eq:errorsize}, we can also make $\norm{P}$ arbitrarily small; therefore, since $\sigma>0$,  we are guaranteed to satisfy this condition asymptotically. Thus, asymptotically, all singular values of $G(\tilde\vt)$ are nonzero and the matrix has full rank.
\end{proof}

Now we consider the solutions to $G^T(\vt')\vec w=\vec\alpha$ as compared to $G^T(\tilde\vt)\tilde{\vec w}=\tilde{\vec\alpha}$. We begin by restating a useful result from Ref.~\cite{cadzow1973finite}, labeled there as Fundamental Theorem 2, with our notation.
\begin{thm}\label{thm:orthogonality}
Given the $k\times d$ matrix $G^T$ with rank $k$ and the $k\times 1$ vector $\vec\alpha$, there exists a $k\times 1$ vector $\vec{v^0}$ such that 
\[
\vec\alpha\cdot\vec{v^0}=\max\limits_{\norm{G\vec v}_1\leq 1}\vec\alpha\cdot\vec v=\max\limits_{\norm{G\vec v}_1= 1}\vec\alpha\cdot\vec v,
\]
and at least $k-1$ components of $G\vec{v^0}$ are zero, that is:
\[
   \vec{g_i}\cdot\vec{v^0}=0\quad \mathrm{ for }\quad  i\in\left[i_1, i_2, \cdots, i_{k-1}\right] \quad \mathrm{ with }\quad 1\leq i_\ell\leq d,
\]
where $\vec{g_i}$ denotes the $i^{\mathrm{th}}$ column of $G^T$. Furthermore, the set of vectors 
\[
\left[\vec{g_{i_1}}, \vec{g_{i_2}}, \cdots, \vec{g_{i_{k-1}}}\right]
\]
is linearly independent.
\end{thm}
\noindent This theorem is about the protocol problem. That is, $\vec v^0$ is the solution vector to the protocol problem. Furthermore, we recall that, by strong duality,  
\begin{equation}
u''=\vec\alpha\cdot\vec{v^0}=\min_{G^T(\vt')\vec w=\vec\alpha}\norm{\vec w}_\infty=u'.
\end{equation}

We now compare how the solution of the protocol problem is perturbed by considering $G^T(\tilde\vt)$ and $\tilde{\vec\alpha}$ as opposed to $G^T(\vt')$ and $\vec\alpha$. To that end, we prove the following theorem.

\begin{thm}\label{thm:wbound}
Consider the linear systems of equations $G^T(\vt')\vec w=\vec\alpha$ and $G^T(\tilde\vt)\tilde{\vec w}=\tilde{\vec\alpha}$, where we recall $\vec\alpha=\nabla q(\vt')$ and $\tilde{\vec\alpha}=\nabla q(\tilde\vt)$. Then
\[
\norm{\tilde{\vec w}}_\infty = \norm{\vec w}_\infty+\mathcal{O}(\vec\Delta),
\]
where $\vec w$ is the solution to the protocol problem with $G(\vt')$ and $\vec\alpha$ and $\tilde{\vec w}$ is the solution to the protocol problem with the approximate $G(\tilde\vt)$ and $\tilde{\vec\alpha}$.
\end{thm}
\begin{proof}
As in Theorem \ref{thm:consistency}, write $G(\tilde\vt)=G(\vt')+P$ with perturbation matrix $P$. Similarly, we define a perturbation vector $\vec p$ such that $\tilde{\vec\alpha}=\vec\alpha+\vec p$. We  expand $q(\tilde\vt)$ about $\vt'$ as
\begin{equation}\label{eq:qexpansion}
    q(\tilde\vt)=q(\vt')+\nabla q(\vt')\cdot\vec\Delta+\dots
\end{equation} 
We then have
\begin{equation}
\begin{split}
    \tilde{\vec\alpha}=\grad q(\tilde\vt)&=\grad q(\vt')+\mathcal{O}(\vec\Delta)=\vec\alpha+\mathcal{O}(\vec\Delta).
    \end{split}
\end{equation}
Therefore we have $\vec p=\mathcal{O}(\vec\Delta)$ and (from Theorem \ref{thm:consistency}) $P=\mathcal{O}(\vec\Delta)$. 

Now, due to strong duality, we conclude that the solution of the dual protocol problem is equal to that of the protocol problem. That is,
\begin{equation}\label{eq:approxw}
    \begin{split}
    \norm{\vec w}_\infty &= \vec\alpha\cdot\vec{v^0},\\ 
    \norm{\tilde{\vec w}}_\infty &= \tilde{\vec\alpha}\cdot(\vec{v^0}+\vec{\epsilon^0}),
    \end{split}
\end{equation}
for the unperturbed and perturbed problems, respectively. We have introduced $\vec\epsilon^0$ as the perturbation in $\vec{v^0}$ in the solution to the dual protocol problem in response to the perturbations $P$ and $\vec p$. 

Consider the solution to the unperturbed problem. We introduce the $k\times (k-1)$ matrix $M^T=\left(\vec{g_{i_1}}, \vec{g_{i_2}}, \cdots, \vec{g_{i_{k-1}}}\right)$ with columns as defined in Theorem \ref{thm:orthogonality}. By the same theorem, the solution vector satisfies
\begin{equation}
    M\vec{v^0}=\vec 0,
\end{equation}
where we note that $M$ is a submatrix of $G(\vt')$. Hence the solution vector must be jointly orthogonal to all columns of $M^T$. Via a $k$-dimensional generalization of the determinant formula for a cross product (obtained from Cramer's rule), we have the unnormalized solution vector
\begin{equation}\label{eq:formaldet}
    \vec{v}= \det \left[\begin{pmatrix}
    \vec{e} \\
    M
    \end{pmatrix}\right],
\end{equation}
where $\vec{e}=(\vec{e_1}, \vec{e_2}, \cdots \vec{e_k})^T$ represents a vector of the standard-basis vectors. Note that $\vec v$ is unique up to scalar multiplication. We then have component-wise
\begin{equation}\label{eq:formaldetcomp}
    \vec{v}_n=(-1)^{n+1}\det(M_n)\vec{e_n}
\end{equation}
where we naturally define $M_n$ as the unique $(k-1) \times (k-1)$ submatrix of $M$ that results from eliminating the first row and $n^\mathrm{th}$ column of the matrix in \cref{eq:formaldet}. We normalize the solution vector to the protocol problem so that it satisfies the condition in Theorem \ref{thm:orthogonality} as
\begin{equation}\label{eq:normsol}
    \vec{v^0}=\frac{\mathrm{sgn}(\vec\alpha\cdot\vec{v})}{\norm{G(\vt')\vec{v}}_1}
    \vec{v}=\mathcal{N}\vec{v},
\end{equation}
where we have implicitly defined the normalization factor $\mathcal{N}$.

We now consider the perturbed problem. Introduce two $k\times(k-1)$ matrices: $\overline{M}^T=\left(\vec{g_{j_1}}, \vec{g_{j_2}}, \cdots, \vec{g_{j_{k-1}}}\right)$, a submatrix of $G^T(\vt')$, and $\overline{Q}^T$, the corresponding submatrix of $P$. That is, $\overline{M}+\overline{Q}$ is a submatrix of $G(\tilde\vt)$. We pick indices $j_\ell\in [1, d]$ in accordance with Thm. \ref{thm:orthogonality} such that the solution vector $\vec{v^0}+\vec{\epsilon^0}$ of the perturbed protocol problem satisfies
\begin{equation}\label{eq:overbareq}
    (\overline{M}+\overline{Q})(\vec{v^0}+\vec{\epsilon^0})=\vec 0.
\end{equation}
Similar to the unperturbed case, Eq. (\ref{eq:overbareq}) has an unnormalized solution vector given by the determinant
\begin{equation}\label{eq:formaldetapprox}
  \vec{v}+\vec{\epsilon}=\det\left[
    \begin{pmatrix}
    \vec{e} \\
    \overline{M}+\overline{Q}
    \end{pmatrix}\right],
\end{equation}
which component-wise reads
\begin{equation} \label{eq:perturbcomp}
    \vec{v}_n+\vec{\epsilon}_n= (-1)^{n+1}\det(\overline{M}_n+\overline{Q}_n)\vec{e_n}.
\end{equation}
As in the unperturbed case, $\overline{M}_n+\overline{Q}_n$ is the submatrix of $\overline{M}+\overline{Q}$ corresponding to eliminating the first row and $n^{\mathrm{th}}$ column of the matrix inside the determinant of Eq.\ (\ref{eq:formaldetapprox}). The corresponding normalized solution vector to the perturbed protocol problem is
\begin{equation}
  \vec{v^0}+\vec{\epsilon^0}=\frac{\mathrm{sgn}(\tilde{\vec\alpha}\cdot(\vec{v}+\vec\epsilon))}{\norm{(G(\vt')+P)(\vec{v}+\vec\epsilon)}_1}(\vec{v}+\vec\epsilon)
=\overline{\mathcal{N}} (\vec{v}+\vec\epsilon),
\end{equation}
where we have implicitly defined the normalization factor $\overline{\mathcal{N}}$.

We now consider several cases for how the unperturbed and perturbed solution vectors are related: (1) $M=\overline{M}$, (2) $M\neq \overline{M}$ and $\overline{M}$ has full rank, (3) $M\neq \overline{M}$ and $\overline{M}$ does \emph{not} have full rank. Intuitively, case (1) corresponds to when the solution vectors of the unperturbed and perturbed protocol problems are orthogonal (via Thm. \ref{thm:orthogonality}) to the same set of columns of $G^T(\vt')$ and $G^T(\tilde\vt)$, respectively. That is, the solution vector in the perturbed case is merely a perturbed version of the solution vector in the unperturbed case. In cases (2) and (3), the set of columns of $G^T$ to which the unperturbed and perturbed solution vectors are orthogonal \emph{do not} have the same indices. Intuitively, this means that the perturbed problem solution vector is not simply the perturbed version of the solution vector in the unperturbed problem. These cases divide into two options. In case (2), this set of columns of $G^T(\tilde{\vt})$ in the perturbed case (given by the rows of $\overline{M}+\overline{Q}$) corresponds to a set of unperturbed columns of $G^T(\vt')$ (given by the rows of $\overline{M}$) that are independent -- i.e. $\overline{M}$ has full rank. In particular, this means that the unperturbed version of the perturbed solution vector is a \emph{candidate} solution to the unperturbed problem. By candidate solution we refer to the fact that any choice of independent columns of $G(\vt')$ could correspond to a \emph{possible} solution vector according to Thm. \ref{thm:orthogonality}, in the sense that any such choice picks out a candidate, unnormalized solution vector via Eq.\ (\ref{eq:formaldet}). There are at most $\binom{d}{k-1}$ such candidate solutions based on picking the set of $k-1$ columns that define a possible $M$. 
In case (3), $\overline{M}$ does not have full rank -- i.e. the corresponding unperturbed columns are not independent and the unperturbed version of the perturbed solution vector is \emph{not} a candidate solution vector to the unperturbed problem.

We now examine the cases one by one in detail and find that we can rule out cases (2) and (3). Starting with case (1), we may drop the bar on $\overline{M}$ as $M=\overline{M}$. We then use a bound on determinants of perturbed matrices from Ref.~\cite{ipsen2008perturbation} (see Remark 2.9 therein) and obtain
\begin{equation}
\abs{\det(M_n)-\det(M_n+\overline{Q}_n)}\leq s_{k-2}\norm{\overline{Q}_n},
\end{equation}
where $s_{k-2}\leq(k-1)\sigma_1\cdots\sigma_{k-2}$ is the $(k-2)^{\mathrm{nd}}$ elementary symmetric function in the singular values $\sigma_1\geq\cdots\geq\sigma_{k-1}$ of $M_i$ \cite{ipsen2008perturbation}. Importantly, $\sigma_1=\norm{M_n}\leq\sqrt{k-1}\norm{M_n}_\infty=\mathcal{O}(1)$ and $\norm{\overline{Q}_n}=\mathcal{O}(\vec\Delta)$ as $\norm{\overline{Q}_n}\leq\sqrt{\sum_{ab}{(\overline{Q}_n})_{ab}^2}$ and all elements of $\overline{Q}_n$ are of size $\mathcal{O}(\vec\Delta)$. Therefore,
\begin{equation}
    \abs{\det(M_n)-\det(M_n+\overline{Q}_n)}= \mathcal{O}(\vec\Delta),
\end{equation}
which directly implies that $\vec\epsilon=\mathcal{O}(\vec\Delta)$. 
Having established that $\vec \epsilon=\mathcal{O}(\vec\Delta)$, we now consider the normalization factors $\mathcal{N}$ and $\overline{\mathcal{N}}$. Recall that
\begin{equation}
    \abs{\overline{\mathcal{N}}}=\frac{1}{\norm{(G(\vt')+P)(\vec{v}+\vec\epsilon)}_1}.
\end{equation}
By the triangle inequality,
\begin{equation}
\begin{split}
    \norm{(G(\vt')\vec{v}}_1-\norm{(G(\vt')\vec\epsilon+P(\vec{v}+\vec\epsilon)}_1 \leq \norm{(G(\vt')+P)(\vec{v}+\vec\epsilon)}_1 \leq \norm{(G(\vt')\vec{v}}_1+\norm{(G(\vt')\vec\epsilon+P(\vec{v}+\vec\epsilon)}_1.
\end{split}
\end{equation}
Then, using a binomial expansion yields
\begin{equation}
    \abs{\mathcal{N}}-\mathcal{O}(\vec\Delta)\leq\abs{\overline{\mathcal{N}}}\leq\abs{\mathcal{N}}+\mathcal{O}(\vec\Delta),
\end{equation}
so $\overline{\mathcal{N}}=\mathcal{N}+\mathcal{O}(\vec\Delta)$, where we note that $\mathcal{N}=\mathcal{O}(1)$.
Therefore, for case (1), the perturbed solution vector is $\vec{v^0}+\mathcal{O}(\vec\Delta)$, and $\vec{\epsilon^0}=\mathcal{O}(\vec\Delta)$. Eq. (\ref{eq:approxw}) then yields
\begin{equation}
    \begin{split}
        \norm{\tilde{\vec w}}_\infty &= (\vec\alpha+\vec p)\cdot(\vec{v^0}+\vec{\epsilon^0})\\
        &=\vec\alpha\cdot\vec{v^0}+\mathcal{O}(\vec\Delta)\\
        &=\norm{{\vec w}}_\infty +\mathcal{O}(\vec\Delta).
    \end{split}
\end{equation}

We now demonstrate that neither case (2) nor case (3) can arise. Starting with case (2), we recall from the discussion above that $\overline{M}$ corresponds to a candidate solution for the original unperturbed problem. The corresponding candidate (unnormalized) solution vector to the unperturbed problem, $\overline{\vec v}$, can be found by an equation analogous to \cref{eq:formaldet}. By the same argument as in case (1), when we perturb $\overline{M}$ by $\overline{Q}$ to obtain the perturbed problem, the unperturbed candidate solution vector $\overline{\vec v}$ may only be perturbed by $\mathcal{O}(\vec\Delta)$ and similarly the corresponding candidate solution value may also only be perturbed by $\mathcal{O}(\vec\Delta)$. 
Let the difference between the candidate solution corresponding to $\overline{M}$ and the true solution to the unperturbed problem be given by $r$. As asymptotically we can make $\vec\Delta$ arbitrarily small we may always make $\norm{\vec\Delta} \ll r$. This contradicts the fact that the solution vector $\overline{M}+\overline{Q}$ is the solution to the perturbed problem as an approach like case (1) is guaranteed to offer a better solution than case (2) for sufficiently small perturbations. Therefore case (2) cannot arise.

Similarly, we can show that case (3) cannot arise. In this case, $\overline{M}$ is rank-deficient and consequently does not correspond to a candidate solution to the original unperturbed problem. Also, due to its rank deficiency, we know there exists a linear combination of rows via $\mathcal{O}(1)$ coefficients such that we may transform $\overline{M}$ using row operations into a form that it has a row of all zeros. Call this transformation $T$. We then have that $T(\overline{M}+\overline{Q})$ has a row with all elements of size $\mathcal{O}(\vec\Delta)$.
Consider $\det(\overline{M}_n+\overline{Q}_n)=\det(T(\overline{M}_n+\overline{Q}_n))$ as in Eq.\ (\ref{eq:perturbcomp}), where we eliminate the $n^{\mathrm{th}}$ column of $T(\overline{M}+\overline{Q})$ to obtain $T(\overline{M}_n+\overline{Q}_n)$. Eliminating this column does not change the fact that $T(\overline{M}_n+\overline{Q}_n)$ has a row with all elements of size $\mathcal{O}(\vec\Delta)$. Consequently, $\det(\overline{M}_n+\overline{Q}_n)=\det(T(\overline{M}_n+\overline{Q}_n))=\mathcal{O}(\vec\Delta)$ and, therefore, all components of the \emph{unnormalized} perturbed problem solution vector must be $\mathcal{O}(\vec\Delta)$.
Let this unnormalized solution vector be $\overline{\vec v}$. From before, we have $\overline{\mathcal{N}}=\mathcal{N}+\mathcal{O}(\vec\Delta)$ and $\tilde{\vec\alpha}=\vec\alpha+\mathcal{O}(\vec\Delta)$, so the solution corresponding to $\overline{\vec v}$ is 
\begin{equation}
    \mathcal{\overline{N}}(\tilde{\vec\alpha}\cdot \overline{\vec v})=
    \mathcal{N}(\vec\alpha\cdot\overline{\vec v})+\mathcal{O}(\vec\Delta)=\mathcal{O}(\vec\Delta).
\end{equation}
In the second equality we used that $\mathcal{N}=\mathcal{O}(1)$ and $\overline{\vec v}=\mathcal{O}(\vec\Delta)$. Furthermore, let $\vec v$ be the unnormalized solution vector for the unperturbed problem, then
\begin{equation}
    \mathcal{N}(\vec\alpha\cdot\overline{\vec v})\leq  \mathcal{N}(\vec\alpha\cdot\vec v),
\end{equation}
which implies that, asymptotically,
\begin{equation}
    \overline{\mathcal{N}}(\tilde{\vec\alpha}\cdot\overline{\vec v})\lesssim  \mathcal{N}(\vec\alpha\cdot\vec v).
\end{equation}
The right-hand side of this inequality is $\mathcal{O}(1)$, whereas the left-hand side is $\mathcal{O}(\vec\Delta)$. Therefore, asymptotically, the perturbed solution is no longer within $\mathcal{O}(\vec\Delta)$ of the solution to the unperturbed problem. As a result, asymptotically, an approach like case (1) is guaranteed to result in a better candidate solution vector than a candidate arising from case (3). Thus case (3) cannot lead to a solution, concluding the proof.

\end{proof}

Theorem \ref{thm:wbound} has immediate consequences that we use in evaluating our protocol performance and in the proofs in \cref{appendix:protocolasymptotics}. It also implies the following useful corollary.
\begin{cor}\label{cor:constantbound}
$\norm{\tilde {\vec w}}$ can,  asymptotically, be upper bounded by a constant.
\end{cor}
\begin{proof}
$\norm{{\vec w}}=\mathcal{O}(1)$, and, asymptotically, $\delta$---which bounds $\norm{\vec\Delta}$ (see \cref{eq:errorsize})---can be made arbitrarily small. This directly implies the result.
\end{proof}

\section{Proof of protocol optimality}\label{appendix:protocolasymptotics}
\sectioninfo{We rigorously demonstrate that the two-step protocol described in \cref{app:analyticprotocol} is asymptotically optimal.} 

Using the results of the previous section, we rigorously demonstrate that the two-step protocol described in \cref{app:analyticprotocol} is optimal. In particular, we focus on the effects of using the estimate $\tilde\vt$ from step 1 of the protocol, as opposed to the true $\vt'$, for step 2 of the protocol and demonstrate that, asymptotically, the errors introduced are negligible.

We begin by sketching how the two-step protocol saturates the MSE bound in \cref{eqn:generalmsebound} and, therefore, yields an optimal estimate of the function $q(\vt')$. We then fill in the details to rigorously obtain the result. The MSE of the full protocol is given by
\begin{equation}
    \label{eqn:biasvariancegeneral}
    \begin{split}
        \mathcal{M}&=\mathbb{E}\left[(\tilde{q} - q(\vt'))^2\right] \\
        &=\mathcal{M}_1+ \mathcal{M}_2,
    \end{split}
\end{equation}
with 
\begin{align}
    \mathcal{M}_1 &= \mathbb{E}_{\tilde\vt}[\mathrm{Var}_{\tilde \lambda}[\tilde \lambda]],\\
    \mathcal{M}_2 &= \mathbb{E}_{\tilde\vt}\left[(q(\tilde\vt)+\lambda-\tilde{\vec\alpha}\cdot\tilde\vt- q(\vt'))^2\right].
\end{align}
The variance of the estimation of $\lambda$, for fixed $\tilde{\vt}$, is $\norm{\tilde{\vec w}}_\infty^2/t_2^2$ \cite{Eldredge2018}.  We then show that
\begin{equation}\label{eq:term1}
   \mathcal{M}_1= \frac{\norm{{\vec w}}_\infty^2}{t_2^2}\left(1+\frac{\mathcal{C}}{t_1}+\mathcal{O}\left(t_1^{-2}\right) \right),
\end{equation}
with some constant $\mathcal{C}$. Given $t_1=t^{p}$ with $p\in(1/2,1)$, we conclude that $\mathcal{M}_1$ is asymptotically given by $\frac{\norm{{\vec w}}_\infty}{t_2^2}\sim \frac{\norm{\vec w}_\infty}{t^2}$. On the other hand, we also show that $\mathcal{M}_2$ is of order $\mathcal{O}(t_1^{-4})$ and is therefore asymptotically negligible. Thus, we asymptotically have
\begin{equation}\label{eq:protocolfinalagain}
    \mathcal{M}\sim\frac{\norm{\vec w}_\infty}{t^2},
\end{equation}
where we recall that $\vec w$ is the optimal weight vector obtained from the protocol problem for $G = G(\vt'), \vec\alpha = \nabla q(\vt')$.

The point of the asymptotics is that the problem we actually solve in practice with our estimate $\tilde \vt$ introduces asymptotically negligible corrections. Via the same proofs as in the linear case, we know that the protocol problem used to obtain $\vec w$ gives a solution equivalent to the corresponding bound problem, and therefore the protocol is asymptotically optimal.
We now fill in the details and derive the asymptotic behavior of $\mathcal{M}_1$ and $\mathcal{M}_2$.

\paragraph{\textbf{Derivation of $\mathcal{M}_1$.}}
An immediate consequence of Theorem \ref{thm:wbound} is 
\begin{equation}\label{eq:bndtildeq}
\begin{split}
    \expect_{\tilde\vt}[\mathrm{Var}_{\tilde{\lambda}}[\tilde{\lambda}]]=\frac{\expect[\norm{\tilde{\vec w}}_\infty^2]}{t_2^2}= \frac{\norm{{\vec w}}_\infty^2}{t_2^2}\expect\left[\left(1+\mathcal{B}\vec\Delta\right)^2\right]
\end{split}
\end{equation}
for some constant $\mathcal{B}$. Note that we can expand
\begin{equation}
\expect[(1 + \mathcal B \norm{\vec\Delta})^2] = 1 + \mathcal B^2\expect[\norm{\vec\Delta}^2] + 2\mathcal B\expect[\norm{\vec\Delta}].
\end{equation}
Since $\expect[\norm{\vec\Delta}^2]$ is the sum of the squared deviations of the individual $\theta_i$, it is $\mathcal O\left(\frac 1{t_1^2}\right)$. Similarly, $\expect[\norm{\vec\Delta}]$ is $\mathcal O\left(\frac 1{t_1}\right)$. As a result, we can expand \cref{eq:bndtildeq} to
\begin{equation}
    \mathcal M_1=\expect_{\tilde\vt}[\mathrm{Var}_{\tilde \lambda}[\tilde \lambda]]= \frac{\norm{{\vec w}}_\infty^2}{t_2^2}\left(1+\frac{\mathcal{C}}{t_1}+\mathcal{O}\left(t_1^{-2}\right) \right)
\end{equation}
for some constant $\mathcal{C}$, in agreement with \cref{eq:term1}. 

\paragraph{\textbf{Derivation of $\mathcal{M}_2$.}} We begin with
\begin{equation}\label{eq:term2app}
    \mathcal{M}_2=\expect_{\tilde\vt}\left[(q(\tilde\vt)+\lambda-\tilde{\vec\alpha}\cdot\tilde\vt- q(\vt))^2\right].
\end{equation}
We define the vector $\vec C(\tilde\vt)$ to store constants which simplify our derivation through $C_i(\tilde\vt)=f_i(\tilde\vt)-\tilde\vt\cdot\nabla f_i(\tilde\vt)$. We then define
\begin{equation}\label{eq:lambdaapp}
    \lambda(\vec f)=\tilde{\vec w}\cdot(\vec f-\vec C(\tilde\vt))
\end{equation}
as the linear function we measure in step 2 of our protocol. Similar to \cref{eq:qexpansion}, we expand $q(\vt)$ about $\tilde\vt$:
\begin{equation}\label{eq:qexpansion2}
    q(\vt)=q(\tilde\vt)-\nabla q(\tilde\vt)\cdot\vec\Delta+\frac{T_2}{2}+\cdots,
\end{equation}
where $T_n=\mathcal{O}(\Delta^n)$. Inserting this expansion into \cref{eq:term2app}, where we use the definition of $\tilde{\vec\alpha}$, we find that $\mathcal M_2$ is equivalent to 
\begin{equation}\label{eq:term2}
\begin{split}
    \expect_{\tilde\vt}\left[\left(\lambda -\tilde{\vec\alpha}\cdot\tilde\vt+\tilde{\vec\alpha}\cdot\vec\Delta-\frac{T_2}{2}+\frac{T_3}{3}-\cdots\right)^2\right]=\expect_{\tilde\vt}\left[\left(\underbrace{\lambda -\tilde{\vec\alpha}\cdot\vt'}_{(*)}-\frac{T_2}{2}+\frac{T_3}{3}-\cdots\right)^2\right].
\end{split}
\end{equation}
To simplify the term labeled $(*)$, we insert \cref{eq:lambdaapp} and  use  $\tilde{\vec\alpha}\cdot\vt'=G^T(\tilde\vt)\tilde{\vec w}\cdot\vt'=\tilde{\vec w}\cdot G(\tilde\vt)\vt'$. This yields
\begin{equation}\label{eq:*}
    \lambda -\tilde{\vec\alpha}\cdot\vt=\tilde{\vec w}\cdot\left(\vec f-\vec C(\tilde\vt)-G(\tilde\vt)\vt'\right).
\end{equation}
Now consider the quantity $\vec f-C(\tilde\vt)-G(\tilde\vt)\vt'$. For the $i^{\mathrm{th}}$ component, we have
\begin{equation}\label{eq:componentexpansion}
    \vec f_i-C_i(\tilde\vt)-[G(\tilde\vt)\vt']_i=(f_i-f_i(\tilde\vt))+\nabla f_i(\tilde\vt)\cdot\vec\Delta,
\end{equation}
where we employed the definitions of $\vec C(\tilde\vt)$ and $G(\tilde\vt)$. Since $f_i(\vt)$ is an analytic function in a $\delta$-ball around $\vt'$, we can expanded it about $\vt'$ as in \cref{eq:phiexpansion}. Substituting this expansion into \cref{eq:componentexpansion} we arrive at
\begin{equation}
    \mathrm{\cref{eq:componentexpansion}}=\left( \nabla f_i(\tilde\vt)-\nabla f_i(\vt')\right)\cdot\vec\Delta-\frac{S_2}{2}+\mathcal{O}(\vec\Delta^3)
\end{equation}
with $S_n=\mathcal{O}(\Delta^n)$. From \cref{eq:phiexpansion}, we conclude that 
\begin{equation}
    \begin{split}
    \left(\nabla f_i(\tilde\vt)-\nabla f_i(\vt')\right)\cdot\vec\Delta=S_2+\mathcal{O}(\vec\Delta^3) = \mathcal{O}(\vec{\Delta}^2).
    \end{split}
\end{equation}
Therefore, \cref{eq:componentexpansion} is of order $\mathcal{O}(\vec{\Delta}^2)$. Furthermore, Theorem \ref{cor:constantbound} implies that, asymptotically, $\norm{\tilde{\vec w}}_\infty$ can be upper bounded by a constant, i.e.~the magnitude of each element of $\tilde{\vec w}$ is upper bounded by a constant. Combining these facts, we find that \cref{eq:*} is of order $\mathcal{O}(\vec{\Delta}^2)$. Together with \cref{eq:term2}, this, in turn, implies that $\mathcal M_2$ is $\mathcal{O}(\vec{\Delta}^4)=\mathcal{O}(t_1^{-4})$.

\section{Review of the protocol by Eldredge \textit{et.~al.}}
\sectioninfo{For completeness, we review one of the protocols for optimal estimation of a linear combination of parameters from Eldredge \textit{et.~al.}, which is used as a subroutine in our optimal protocol. } 

In this section, we briefly summarize one of the protocols for optimal estimation of a linear combination of parameters from Eldredge \textit{et.~al.}~\cite{Eldredge2018}, because it is as a subroutine in our protocols to obtain an estimate of Eq.\ (\ref{eq:lambda}) and Eq.\ (\ref{eqn:lambdageneral}) in the main text with variances given by Eq.\ (\ref{eqn:linearprotocolvariance}) and Eq.\ (\ref{eq:protocolfinal}), respectively. We seek to measure a linear combination (up to a constant shift) of the form 
\begin{equation}
\lambda=\vec w\cdot\vec f.    
\end{equation}
Several specific protocols to obtain the optimal MSE estimate of such a linear combination are given in Ref. \cite{Eldredge2018}. Here we present the first and simplest such protocol.

In this protocol, we suppose to have access to a time-dependent control over our evolution. We begin with a $d$-qubit GHZ input state of the quantum sensors given by 
\begin{equation}
    \ket{\psi_0}=\frac{1}{\sqrt{2}}\left(\ket{0}^{\otimes d}+\ket{1}^{\otimes d}\right).
\end{equation}
Under evolution by $\hat\sigma^z$ as in Eq.\ (\ref{eq:H}) in the main text, each qubit sensor accumulates a relative phase between the $\ket{0}$ and $\ket{1}$ states. We perform a partial time evolution so that each qubit sensor is evolved for a time proportional to the corresponding weight $w_i$ on that sensor. We realize this by applying $\hat\sigma_i^x$ to the $i^{\mathrm{th}}$ qubit at time $t_i=t(1+w_i)/2$. This results in an effective evolution of our state by the unitary
\begin{equation}
\hat{U}(t)=e^{-i\frac{t}{2}\sum_{i=1}^d w_if_i(\vt)\hat\sigma_i^z }.    
\end{equation}
We note that this scheme assumes that $w_i\in[-1,1]$ and that the largest $|w_i|$ is equal to 1. We can always achieve this by rescaling the vector $\vec w$. 
Under this unitary evolution, the final state of the qubits is
\begin{equation}
    \ket{\psi_f}=\frac{1}{\sqrt{2}}\left(e^{-it\lambda/2}\ket{0}^{\otimes d}+ e^{it\lambda/2}\ket{1}^{\otimes d}\right).
\end{equation}
We then make a measurement of the overall parity of the state using $\hat P=\bigotimes_{i=1}^d \hat{\sigma}_i^x$. Note that this measurement can be performed locally at each site. Furthermore, measurement of the expectation value $\langle\hat P\rangle(t)$ allows for estimation of $\lambda$ with the optimal accuracy given by Eq.\ (\ref{eqn:linearprotocolvariance}) \cite{wineland1994squeezed}. 

\let\oldaddcontentsline\addcontentsline
\renewcommand{\addcontentsline}[3]{}
\let\addcontentsline\oldaddcontentsline

\end{document}